\documentclass[reqno,11pt]{amsart}
\usepackage{color}
\IfFileExists{mymtpro2.sty}{%
  \usepackage[subscriptcorrection]{mymtpro2}
}{}

\usepackage{a4,amssymb}
\usepackage{graphicx}

\newtheorem{theorem}{Theorem}

\newtheorem{thm}{Theorem}[section]
\newtheorem{lemma}{Lemma}[section]

\newtheorem{pr}{Proposition}[section]
\theoremstyle{definition}

\newcommand{\bel}{\begin{equation} \label}
\newcommand{\ee}{\end{equation}}

\theoremstyle{remark}

\newtheorem{myremarks}[theorem]{Remarks}
\newtheorem{myexample}[theorem]{Example}

    {\end{nummer}\end{myremarks}}

\newcounter{numcount}
\newcommand{\labelnummer}{\mbox{\normalfont (\roman{numcount})}}%

\makeatletter

\newenvironment{nummer}%
  {\let\curlabelspeicher\@currentlabel%
    \begin{list}{\labelnummer}%
      {\usecounter{numcount}\leftmargin0pt%
        \topsep0.5ex\partopsep2ex\parsep0pt\itemsep0ex\@plus1\p@%
        \labelwidth2.5em\itemindent3.5em\labelsep1em%
      }%
    \let\saveitem\item%
    \def\item{\saveitem%
      \def\@currentlabel{{\upshape\curlabelspeicher}$\,$\labelnummer}}%
    \let\savelabel\label%
    \def\label##1{\savelabel{##1}%
      \@bsphack%
        \ifmmode\else%
          \protected@write\@auxout{}%
          {\string\newlabel{##1item}{{\labelnummer}{\thepage}}}%
        \fi%
      \@esphack%
    }%
  }{\end{list}}%

\renewcommand{\appendix}{\def\thesection{\textsc{Appendix}}}

 \let\leq\le
 \let\geq\ge

\DeclareMathOperator{\tr}{tr\kern1pt}

\newcommand{\N}{{\mathbb N}}

\makeatletter

\newif\ifper\pertrue
\def\per{.}

\def\bti{\@ifnextchar[\bbti\bbbti}
\def\bbti[#1]#2{#2, #1.}
\def\bbbti#1{#1.}

\def\z{\@ifnextchar[\zz\zzz}
\def\zz[#1]#2#3#4#5{\perfalse\emph{#2} \textbf{#3}, #4 (#5) [#1]}
\def\zzz#1#2#3#4{\emph{#1} \textbf{#2}, #3 (#4)\ifper\per\fi\pertrue}

\def\pub{\@ifstar\pubstar\pubnostar}
\def\pubnostar{\@ifnextchar[\@@pubnostar\@pubnostar}
\def\@@pubnostar[#1]#2#3#4{#2, #3, #4, #1\ifper\per\fi\pertrue}
\def\@pubnostar#1#2#3{#1, #2, #3\ifper\per\fi\pertrue}
\def\pubstar[#1]#2#3#4{\perfalse #2, #3, #4 [#1]\pertrue}

\makeatother

\newcommand{\beq}{\begin{equation}}
\newcommand{\eeq}{\end{equation}}
\newcommand{\ba}{\begin{array}}
\newcommand{\ea}{\end{array}}
\newcommand{\bea}{\begin{eqnarray}}
\newcommand{\eea}{\end{eqnarray}}
\newcommand{\beas}{\begin{eqnarray*}}
\newcommand{\eeas}{\end{eqnarray*}}

{

\newcommand{\R}{\mathbb{R}}

\newcommand{\Z}{\mathbb{Z}}

\newcommand{\C}{\mathbb{C}}

\def\P{I\kern-.30em{P}}
\def\E{I\kern-.30em{E}}
\renewcommand{\E}{\mathbb{E}\mkern2mu}
\renewcommand{\P}{\mathbb{P}}

\begin{document}

\title[Eigenvalues statistics for non rank one perturbations]{Eigenvalue statistics for random Schr\"odinger operators with non rank one perturbations}

\author[P.\ D.\ Hislop]{Peter D.\ Hislop}
\address{Department of Mathematics,
    University of Kentucky,
    Lexington, Kentucky  40506-0027, USA}
\email{hislop@ms.uky.edu}

\author[M.\ Krishna]{M.\ Krishna}
\address{Institute for Mathematical Sciences,
IV Cross Road, CIT Campus,
 Taramani, Chennai 600 113,
 Tamil Nadu, India}
\email{krishna@imsc.res.in}

\thanks{PDH was partially supported by NSF through grant DMS-1103104.
MK was partially supported by IMSc Project 12-R\&D-IMS-5.01-0106, and thanks Minami, Dhriti and Anish for discussions on eigenvalue statistics}


\begin{abstract}
   We prove that certain natural random variables associated with the local eigenvalue statistics for generalized lattice Anderson models constructed with finite-rank perturbations are compound Poisson distributed. This distribution is characterized by the fact that the L\'evy measure is supported on at most a finite set determined by the rank. The proof relies on a Minami-type estimate for finite-rank perturbations. For Anderson-type continuum models on $\R^d$, we prove a similar result for certain natural random variables associated with the local eigenvalue statistics. We prove that the compound Poisson distribution
   associated with these random variables has a L\'evy measure whose support is at most the set of positive integers.
 \end{abstract}

\maketitle \thispagestyle{empty}


\tableofcontents

\vspace{.2in}



\section{Statement of the problem and results}\label{sec:introduction}
\setcounter{equation}{0}

We consider random Schr\"odinger operators $H^\omega = L + V_\omega$ on the lattice Hilbert space $\ell^2 (\Z^d)$ (or, for matrix-valued potentials, on $\ell^2 (\Z^d) \otimes \C^{m_k}$),
or on the Hilbert space $L^2(\R^d)$, and prove that certain natural random variables
associated with the local eigenvalue statistics around an energy $E_0$, for energies $E_0$ in the region of complete localization, are distributed according to a compound Poisson distribution.
The operator $L$ is the discrete Laplacian on $\Z^d$ or the usual Laplacian on $\R^d$.
For lattice models, the random potential $V_\omega$ has the form
\beq\label{eq:potential1}
(V_\omega f)(j) = \sum_{i \in {\mathcal J}} \omega_i  (P_i f)(j),
\eeq
where $\{ P_i \}_{i \in {\mathcal{J}}}$ is a family of finite-rank projections with the same rank $m_k \geq 1$ and
so that $\sum_{i \in {\mathcal{J}}} P_i = I$. For the models on $\R^d$, the random potential is Anderson-type
\beq\label{eq:potential1-cont}
(V_\omega f)(x) = \sum_{i \in {\Z^d}} \omega_i  u(x-i) f(x),
\eeq
where $u \geq 0$ is a bounded single-site potential of compact support (see, for example, the description in \cite{CHK2}).

In either situation, the coefficients $\{ \omega_i \}$ are a family of independent, identically distributed (iid) random variables with a bounded density of compact support on a product probability space $\Omega$ with probability measure $\P$.

One example on the lattice is the polymer model for which $P_i = \chi_{\Lambda_k(i)}$ is the characteristic function on the cube of side length $k$ so the rank of $P_i$ is $k^d$ and the set $\mathcal{J}$ is chosen so that $\cup_{i \in \mathcal{J}} {\Lambda_k(i)} = \Z^d$.
Another example is a matrix-valued model for which $P_i$, $i \in \Z^d$, projects onto an $m_k$-dimensional subspace,
 and $\mathcal{J} = \Z^d$. The corresponding Schr\"odinger operator is
\beq\label{eq:model1}
H^\omega = L +  \sum_{i \in {\mathcal{J}}} \omega_i  P_i ,
\eeq
where $L$ is the discrete lattice Laplacian $\Delta$ on $\ell^2 (\Z^d)$,
or $\Delta \otimes I$ on $\ell^2(\Z^d) \otimes \C^{m_k}$, respectively.

The one-dimensional dimer model \cite{debievre-germinet1} consists of $d=1$, $k=2$, the set $\mathcal{J} = 2 \Z$, the even integers, and $\Lambda_2 (i)$ consists of the pair of lattice points $\{ i, i+1 \}$.
The rank of the projector $P_i$ is $m_k=2$.

Our results hold for the energy $E_0$ belonging to the domain of complete localization $\Sigma_{\rm CL}$.
The region of complete localization for $H^\omega$ is a closed subset of the almost sure spectrum characterized by dense pure point spectrum and exponentially decaying eigenfunctions. Most importantly, the Green's function at
energies in $\Sigma_{\rm CL}$ exhibits exponential decay. We refer the reader to Appendix A of \cite{CGK2} for a
concise description of $\Sigma_{\rm CL}$, and the to references in that paper for more details.

Our main result on the eigenvalue statistics, Theorem \ref{thm:main1} for the lattice case,
is as follows. Let an energy $E_0$ be in the regime of complete localization $\Sigma_{\rm CL}$.
We define the rescaled, local, eigenvalue point process associated with finite volume restrictions $H_\omega^\Lambda$ of $H^\omega$ as follows. Let $\{ E_j^\omega (\Lambda) \}$ be the eigenvalues of the local Hamiltonian $H_\omega^\Lambda$. For a bounded function $f$ of compact support, we define
\beq\label{eq:pt-process1}
\xi_\Lambda^\omega (f) := \sum_j f(|\Lambda| (E_j^\omega(\Lambda) - E_0)) .
\eeq
The limit points $\xi^\omega$ of this process as $|\Lambda| \rightarrow \infty$ exist and are
infinitely-divisible point processes on $\R$.
We prove that for any bounded Borel subset $I \subset \R$, the associated random variable $\xi^\omega(I)$
is distributed according to a compound Poisson distribution.
Furthermore, the L\'evy measure associated with the characteristic function for $\xi^\omega(I)$
has support in the finite set $\{ 1, 2, \ldots, m_k \}$. When the rank of $P_j$ is one, one recovers a Poisson distribution.

%

We are also able to study the distribution of the random variables $\xi^\omega(I)$ for random Schr\"odinger operators on $\R^d$ for energies in the regime of complete localization. In our main result for $\R^d$, Theorem \ref{thm:main-continum1}, we prove that these random variables also have a compound Poisson distribution. The proof is based on the Wegner estimate and localization.
Since we do not have a Minami-type estimate in this case, the most that we can prove is that the L\'evy measure associated with the characteristic function for $\xi^\omega(I)$ has support in the positive integers $\N$.

We give examples of random operators with compound Poisson and strictly
non-Poisson local statistics at the end of the paper. These examples make it
clear that spectral multiplicity plays a role in
local statistics, in addition to the nature of the spectrum (see also the paper of Naboko, Nichols, and Stolz \cite{naboko-nichols-stolz} and the discussion in section \ref{sec:non-poisson}). The Poisson local
eigenvalue statistics
in the usual Anderson models at high disorder seems
to come from exponential localization and simplicity of the spectrum. Indeed,
the simplicity of the spectrum for lattice models is proved using localization bounds and a Minami estimate \cite{klein-molchanov}.
We strongly suspect that the generalized Anderson-type models on $\Z^d$ with finite-rank perturbations considered
here do not have simple spectrum. This might also be the case for the random Schr\"odinger operator on $\R^d$ ($d>1$)
considered in section \ref{sec:continuous1}.

To our knowledge, these are the first results on eigenvalue statistics for general higher-rank perturbations on the lattice, and the first results on eigenvalue statistics for random Schr\"odinger operators on $\R^d$.

A compound Poisson process is a type of an infinitely-divisible, compound point process.
A typical example is constructed from a family $\{ X_j ~|~ j=0,1,2, \ldots \}$ of independent, identically distributed random variables and a Poisson process $N(t)$ with intensity $\lambda$ that is independent of the $X_j$. Then the process $Y(t) := \sum_{j=0}^{N(t)} X_j$ is a compound Poisson point process. The random variable $Y := Y(1)$, for example, has a distribution function given by
$$
F_Y(w) = \sum_{j=0}^\infty (f_X \ast f_X \ast \cdots \ast f_X )(w) \frac{ \lambda^j e^{-\lambda}}{j!} ,
$$
where $f_X$ is the common distribution function of the random variables $X_j$ and the convolution is taken $j$-times.
Alternately, the random variable $Y$ is distributed according to a compound Poisson distribution if its characteristic function
has the following form:
$$
\E \{ e^{it Y} \} = e^{ \int ( e^{itx} - 1) dM (x) } ,
$$
for some measure $M$ on $\R$ called the L\'evy measure.

There are two works related to our results. F.\ Nakano \cite{nakano} proved that the limit points of the eigenvalue point process for the continuum models studied here are infinitely-divisible point processes. Our result goes beyond this showing that certain natural random variables have compound Poisson distributions.
  In a recent paper, Tautenhahn and Veseli\'c \cite{tautenhahn-veselic} consider the $L+V_\omega$ on $\ell^2 (\Z^d)$, where $L$ is the discrete Laplacian, and $V_\omega(j) = \omega_j + \kappa \sum_{k \in \Z^d} \omega_k v(j-k)$. The function $v: \Z^d \rightarrow \R$ has compact support and $\kappa > 0$ is sufficiently small. In this sense, the potential $V_\omega$ is a small perturbation of the rank one case. Under some additional assumptions, they prove that the local eigenvalue point process is Poisson in the regime of complete localization.
  Recently, F.\ Klopp indicated to us that the methods of C.\ Shirley in
 \cite{shirley} (see also \cite{klopp}) can be used to prove Poisson statistics for the dimer
 model, described above, on the one-dimensional lattice.

 %


Local eigenvalue statistics in the localization regime for random Schr\"odinger operators on the lattice $\Z^d$ have been studied
by Molchanov \cite{molchanov} (Russian-school model on $\R$), by Minami \cite{minami1} and by Combes, Germinet and Klein \cite{CGK1}, ((lattice models, any dimension, localization regime), Killip and Nakano \cite{killip-nakano} (joint eigenvalue and center of localization distribution for lattice models, any dimension). In these case, the limiting eigenvalue point process obtained from $\xi_L$ is proved to be a Poisson point process. The Minami estimate is crucial in proving this result. The local eigenvalue spacing statistics for lattice models in the localization regime was proved by Germinet and Klopp \cite{germinet-klopp}.
Aizenman and Warzel \cite{aizenman-warzel} considered local eigenvalue statistics associated with the Anderson model
on regular rooted trees like the Bethe lattice. Although the random lattice Schr\"odinger operator
has both pure point and absolutely continuous spectrum, the local eigenvalue statistics are always Poissonian. This is explained by the fact that the canopy graph operator, that the authors show is the relevant operator for the study of local eigenvalue statistics, always has pure point spectrum. Thus the local eigenvalue statistics should be Poissonian despite the mixed spectral-type of the original random Schr\"odinger operator on the regular rooted tree.

Almost simultaneously with the present work, Dolai and Krishna \cite{dolai-krishna}
 proved Poisson statistics for the Anderson model on $\ell^2 (\Z^d)$ with $\alpha$-H\"older
continuous single-site probability distribution.
Their work shows that the spacing of
eigenvalues of the finite boxes $\Lambda$, which is like $|\Lambda|^{-1}$
for the case of an absolutely continuous distribution, changes to $|\Lambda|^{-\frac{1}{\alpha}}$ with the singularity of the
distribution, becoming smaller for more singular measures.


\subsection{Contents}

We first give the proof of Proposition \ref{proposition:minami1}, an extended Minami-type estimate for finite-rank perturbations,
in section \ref{sec:minami1}. We then prove that the random variables $\xi^\omega(I)$ obtained from a limiting point process are compound Poisson distributed by studying the characteristic functions in section \ref{sec:processes1}. We study continuum models in section \ref{sec:continuous1}. For these models on $\R^d$,
we show that the Wegner estimate and localization suffice to identify the distribution of $\xi^\omega (I)$
as a compound Poisson distribution and to conclude that the associated L\'evy measure has support in the set $\N$.
We conclude in section \ref{sec:non-poisson} with some examples of random operators for which the distribution of $\xi^\omega(I)$ is a nontrivial compound Poisson distribution. We present a technical result on the convergence of point measures in the appendix.


\section{Preliminaries for the finite-rank lattice models}\label{sec:prelims-lattice1}
\setcounter{equation}{0}

The proof of Theorem \ref{thm:main1} is based on the following observation. For $\Lambda \subset \Z^d$, a suitable cube (see section \ref{sec:minami1}), we denote by $H_\Lambda^\omega$ the restriction of $H^\omega$ to $\Lambda$ with self-adjoint boundary conditions.
Suppose $m_k$ denotes the constant rank of the projectors $P_i$ appearing in \eqref{eq:potential1}.
We consider the following set
\bea\label{eq:smallset}
S_{m_k}(I) = \{\omega:  Tr(E_I(H_\Lambda^\omega)) > m_k \} \subset \Omega.
\eea
We let $\chi_{S_{m_k}(I)}(\omega)$ be the characteristic function of this set and note that
$$
\chi_{S_{m_k}(I)}(\omega) < \frac{1}{m_k} \chi_{S_{m_k}(I)}(\omega) Tr(E_I(H_\Lambda^\omega)),
$$
and that
$$
\chi_{S_{m_k}(I)}(\omega) \leq  \chi_{S_{m_k}(I)}(\omega) ( Tr(E_I(H_\Lambda^\omega)) - m_k),
$$
since on the set $S_{m_k}(I)$, the quantity
$\left(  {\rm Tr} E_I (H_\Lambda^\omega) - m_k \right)$ is bigger than or equal to $1$.
It then follows that
\bea\label{eq:minami1}
\P \{ S_{m_k}(I) \}  & = & \E \{ \chi_{S_{m_k}(I)} \} \nonumber \\
 &< &  \frac{1}{m_k} \E \left(  {\rm Tr} E_I (H_\Lambda^\omega)
 \left(  {\rm Tr} E_I (H_\Lambda^\omega) - m_k \right) \chi_{S_{m_k}(I)} \right)
  \eea
for suitable boxes $\Lambda$.
We then estimate the right side of \eqref{eq:minami1} using spectral
averaging and the lattice argument of Combes-Germinet-Klein \cite{CGK1}.
In this context, we prove the following variant of the Minami estimate. When $m_k = 1$, this is the usual Minami estimate.

\begin{pr}\label{proposition:minami1}
Let $H_\Lambda^\omega$ be the restriction to the cube $\Lambda \subset \Z^d$ of the finite-rank
Anderson model described in \eqref{eq:model1} with uniform rank $m_k \geq 1$. Let $I = [a,b] \subset \R$ be a
finite energy interval. There exists a finite constant $C_M > 0$, depending only on $b,d$ and $\| \rho \|_\infty$, so that
\beq\label{eq:minami-limit1}
\P \left\{ S_{m_k}(I) \right\} \leq {C_M}  |I|^2 |\Lambda|^2 .
\eeq
\end{pr}
We conclude this section with two simple results for our models that will be used in the sequel.

\subsection{The density of states}\label{subsec:dos1}

We compute the density of states (DOS) for the Hamiltonians \eqref{eq:potential1} where the rank of $P_j$ is $m_k$. The trace is on the Hilbert space $\ell^2 (\Lambda)$ or $\ell^2 (\Lambda) \times \C^{m_k}$, according to the model. For suitable cubes $\Lambda \subset \Z^d$, the indices $j$ range over $\mathcal{J} \cap \Lambda$ or $\Lambda$, respectively.
In either case, we write $1_\Lambda = \sum_{j \in \mathcal{J} \cap \Lambda } P_j$ so that
$$
{\rm Tr} E_I(H_\Lambda^\omega) = \sum_{j \in \mathcal{J} \cap \Lambda} {\rm Tr} E_I(H_\Lambda^\omega P_j).
$$
For example, for the case of \eqref{eq:model1}, the projector $P_j$ is the characteristic function $\chi_{\Lambda_k(j)}$
and we have
$$
{\rm Tr} E_I(H_\Lambda^\omega) = \sum_{j=1}^{|\Lambda|/ m_k} {\rm Tr} E_I(H_\Lambda^\omega \chi_{\Lambda_k(j)}).
$$
The following result follows from known results on the integrated density of states (IDS),
see, for example, \cite{CHK1,krishna-stollmann}.
Let us assume that the probability measure for $\omega_0$ has a bounded density with compact support.
Let $\Lambda_L (j)$, respectively, $\Lambda_L$, denote a cube of side length $L > 0$ centered at $j \in \Z^d$, respectively, the origin.
\begin{lemma}\label{lemma:dos1}
For an interval $I = [a,b]$, the limit
\beq
\lim_{L \rightarrow \infty} \frac{1}{| \Lambda_L|} {\rm Tr} \{ E_I (H_{\Lambda_L}^\omega) \},
\eeq
exists almost surely.
For the first case with finite-rank projector $P_j \otimes I$, the almost sure limit is
\beq\label{eq:dos-general1}
N(I) := \E \{ {\rm Tr} ( (P_0 \otimes I )E_I (H^\omega ) ) \} ,
\eeq
and when $P_j = \chi_{\Lambda_k (j)}$, the almost sure limit is
\beq\label{eq:dos-characteristic1}
N(I) := \frac{1}{m_k} \E \{ {\rm Tr}
( E_I (H^\omega) \chi_{\Lambda_k(0)}) \} .
\eeq
Furthermore, $N (E) := N((-\infty, E])$ is Lipschitz continuous. The corresponding DOS $n_{m_k}(E)$
exists as a locally bounded function.
\end{lemma}

\subsection{The Wegner estimate}\label{subsec:wegner1}

The Wegner estimate is well-known for lattice models. The standard proof (see, for example, \cite[appendix]{CGK1})
may be applied to the finite-rank models considered here.

\begin{lemma}\label{lemma:wegner1}
For an interval $I = [a,b]$, there is a constant $c>0$, depending on $b$ and independent of $m_k$ so that
\beq
\E \{ {\rm Tr} \chi_{\Lambda_k} E_I (H^\omega)  \} \leq c |I| | \Lambda_k|  .
\eeq
\end{lemma}


\section{The Minami estimate for finite-rank lattice models}\label{sec:minami1}
\setcounter{equation}{0}
We extend the usual Minami estimate for rank one perturbations to finite-rank perturbations.


\subsection{Rank $m_k$ perturbation bound}\label{subsec:rank-pert1}

Let $\Lambda \subset \Z^d$ be a box and suppose $\Lambda_k (j) \subset \Lambda$. We write $\omega = (\omega_j^\perp, \omega_j)$ to denote the
random variables in $\Lambda$ decomposed relative to $j$. We then have a standard result that a perturbation of rank $m_k$ can change at most $m_k$ eigenvalues of $H_\Lambda^\omega$. For any $\tau_j > 0$, we have
\beq\label{eq:rank-pert1}
| {\rm Tr} E_I (H_\Lambda^\omega) - m_k | \leq {\rm Tr} E_I (H_\Lambda^{(\omega_j^\perp, 0)} + \tau_j P_j)
= {\rm Tr} E_I (H_\Lambda^{(\omega_j^\perp, \tau_j)}) .
\eeq
See, for example, section 4.3 of \cite{CHK2}.

\subsection{Spectral averaging}\label{subsec:sp-ave1}

Let $0 \leq m < M < \infty$. For a finite rank projector $P_j$, we have the spectral averaging result:
$$
(M-m)^{-1} ~\int_m^M ~d \omega_j \langle P_j  \varphi, E_I (H_\Lambda^{(\omega_j^\perp, \omega_j)}) P_j \varphi \rangle
\leq C_1 \| \varphi \|_{\ell^2 (\Z^d)} |I|,
$$
for any $\varphi \in \ell^2 (\Z^d)$. The constant $C_1 > 0$ is independent of $|\Lambda|$.


\subsection{Proof of the extended Minami estimate}\label{subsec:minami-pv1}

We take $\tau_j > M$ and let $\Lambda \subset \Z^d$ be a cube. For each $j \in \mathcal{I}$, let $\{ \varphi_m^{(j)} \}_{m=1}^{m_k}$
 be an orthonormal basis for the range of $P_j$ containing an integer number of cubes $\Lambda_k(j)$ of size $m_k$.

\bea\label{eq:pv-minami1}
 \lefteqn{ \E \{ | {\rm Tr} E_I (H_\Lambda^\omega) (  {\rm Tr} E_I (H_\Lambda^\omega)  - m_k )| \} }&& \nonumber \\
  &=& \sum_{j \in \mathcal{J}} \sum_{m=1}^{m_k}  \E \left\{ \langle \varphi_m^{(j)}, E_I (H_\Lambda^{(\omega_j^\perp, \omega_j)} ) \varphi_m^{(j)}
   \rangle  (  |{\rm Tr} E_I (H_\Lambda^\omega ) - m_k | ) \right\} \nonumber \\
  & \leq & \sum_{j \in \mathcal{J}} \E \left\{ \langle \varphi_m^{(j)}, E_I(H_\Lambda^{(\omega_j^\perp, \omega_j)}) \varphi_m^{(j)} \rangle ~{\rm Tr} E_I (H_\Lambda^{(\omega_j^\perp, \tau_j)} ) \right\} \nonumber \\
   & = & (m_k |I| ~\alpha^{-1}) ~ \sum_{j \in \mathcal{J}} ~\E_{\omega_j^\perp} \left[ \int_M^{M+\alpha} ~d \tau_j ~ \left\{
    {\rm Tr} E_I (H_\Lambda^{(\omega_j^\perp, \tau_j)} ) \right\} \right] \nonumber \\
    & \leq & C_M m_k |I|^2 {|\Lambda|^{2}},
    \eea
    where $C_M > 0$ is a constant depending on $C_1, b, d$, and the density $\rho$.
 To pass from the second to the third line, we take the expectation of the  inner product with respect to $\omega_j$ since the trace does not depend on this random variable. We then integrate with respect to $(\tau_j, \omega_j^\perp)$ which is a new expectation and use the Wegner estimate to bound the trace.


\section{Eigenvalue point processes for finite-rank lattice models}\label{sec:processes1}
\setcounter{equation}{0}

In this section, we prove that the random variable $\xi^\omega (I)$ is distributed according to a compound Poisson distribution.
Recall that this means that the characteristic function has the form
\beq\label{eq:compound-defn1}
\E \left\{ e^{i t \xi^\omega (I) } \right\} = e^{  \int ( e^{itx} - 1) ~dM(x)} .
\eeq
In our case, we will show that the L\'evy measure $M$ is a point measure with support in the finite set $\{ 1 , \ldots , m_k \}$
and with weights as described in Theorem \ref{thm:main1}.

We need some local operators. For any integer $L>0$, so that $m_k$ divides $L$, we take $\Lambda_L$ to be a cube of side length $2L+1$, so that $\Lambda_L := \{ n \in \Z^d ~|~ |n| \leq L \}$.
Let $\beta_L : = (2L+1)^d$. We choose another integer $0 < \ell < L$ so that $N_L$ cubes of side length $2 \ell + 1$ form a non-overlapping cover of $\Lambda_L$ with centers $\{ n_p ~|~ p = 1, \ldots, N_L \}$:
$$
\Lambda_L = \cup_{p=1}^{N_L} \Lambda_\ell (n_p), ~~{\rm with} ~~N_L = [(2L+1)(2 \ell+1)^{-1}]^d.
$$

We define the local Hamiltonian $H_L^\omega := \chi_{\Lambda_L} H^\omega \chi_L$,
the restriction of $H^\omega$ to the cube $\Lambda_L$. For each $p = 1, \ldots, N_L$, we likewise define local operator $H_{p,\ell}^\omega := \chi_{\Lambda_\ell(n_p)} H^\omega \chi_{\Lambda_\ell(n_p)}$. Let $E \in \Sigma_{\rm{CL}}$ and set $I = [a,b] \subset \Sigma_{\rm CL}$ be a finite interval. Let $\tilde{I} := \beta_L^{-1} I + E$ be the scaled energy interval centered at $E$ where $\beta_L = L^d$.
We define two eigenvalue point processes associated with each local operator and the interval ${I}$:
\beq\label{eq:process1}
\xi_{L}^\omega ({I}) := {\rm Tr} ( \chi_{\Lambda_L} E_{\beta_L(H^\omega - E)} (I) )
= {\rm Tr} ( \chi_{\Lambda_L} E_{H_L^\omega}(\tilde{I}) ),
\eeq
and
\beq\label{eq:process2}
\eta_{\ell,p}^\omega ({I}) := {\rm Tr} ( \chi_{\Lambda_\ell(n_p)} E_{\beta_L(H_{p,\ell}^\omega - E)} (I) )
= {\rm Tr} ( \chi_{\Lambda_\ell (n_p)} E_{H_{p,\ell}^\omega} (\tilde{I}) ).
\eeq

In order to state the main result, we need the following limiting behavior of the processes $\xi_L({I})$.
The proof is given in section \ref{subsubsec:proof-main1}.


\begin{lemma}\label{lemma:fixed-prob1}
Let $I$ be a bounded Borel set and let $E$ be such that
the density of states $n_{m_k}(E) \neq 0$, then we have
\beq\label{eq:fixed-prob2}
\E \{ \xi_L^\omega ({I}) \} \leq c |I| ,
\eeq
and
\beq\label{eq:fixed-prob0}
\P \{ \xi_L^\omega ({I}) = j \} \leq \frac{c |I|}{j}, ~~j = 1,2, \ldots .
\eeq
Furthermore, we have
\beq\label{eq:fixed-prob01}
\lim_{L\rightarrow \infty} {\mathbb E}(\xi_L^\omega({I})) = n_{m_k}(E) |I|,
\eeq
and for any $j \geq m_k+1$, we have
\beq\label{eq:fixed-prob1}
\lim_{L \rightarrow \infty} \P \{ \xi_L^\omega ({I}) = j \} = 0 .
\eeq
Finally, there exists a sequence $L_n \rightarrow \infty$ so that
for any $j \in \{ 1, 2, \ldots, m_k \}$ and any compact $I \subset \R$,
the limits
\beq\label{eq:fixed-prob3}
\lim_{n \rightarrow \infty} \P \{ \xi_{L_n}^\omega ({I}) = j \} := p_j(I) ,
\eeq
exist. The weights $p_j(I)$ satisfy $p_j(I) \leq n_{m_k}(E) |I| j^{-1}$,
and at least one of the $p_j(I)$ is non-zero.
\end{lemma}

Our main theorem for the lattice case is the following result.

\begin{thm}\label{thm:main1}
Let $H^\omega$ be a generalized Anderson model on $\Z^d$ with projections $P_j$ having uniform
rank $m_k \geq 1$ as in \eqref{eq:potential1} and \eqref{eq:model1}. Let $E \in \Sigma_{\rm{CL}}$
such that the density of states $n_k(E) >0$.
Let $\xi^\omega$ be a limit point of the eigenvalue point process
 $\xi_L^\omega$ defined in \eqref{eq:process1}.
For each bounded interval $I \subset \R$, the random variables $\xi^\omega(I)$
are compound Poisson random variables
with characteristic function
\beq\label{eq:conv-process11}
\E \left\{ e^{i t \xi^\omega (I)} \right\} = e^{\sum_{j=1}^{m_k} ( e^{itj} - 1) p_j(I)} ,
\eeq
where $p_j(I)$ is defined in \eqref{eq:fixed-prob3}.
Hence, the random variable $\xi^\omega (I)$ has a
compound Poisson distribution with associated L\'evy measure supported on the set $\{ 1, \ldots, m_k \}$ with weights $p_j(I)$.
\end{thm}

We remark that in the above theorem we cannot show that $p_j(I)$ is not zero for any $j \neq 1$, although we suspect this to be true in many cases.
It is not hard to see, using the fact that a Poisson random variable has the same mean and variance,
that if $p_j(I) \neq 0$ for some $j \neq 1$, then the random variable is \textbf{not} Poisson.
In section \ref{sec:non-poisson}, we provide examples of random operators for which this random variable $\xi^\omega (I)$ is
compound Poisson distributed since for these examples we show that $p_j(I) \neq 0$ for some $j > 1$.

The proof of this theorem uses the localization condition. Localization will allow us
to analyze the limit of the processes $\xi_L^\omega (\tilde{I})$
using the family of independent point processes $\{ \eta_{\ell,p}^\omega (\tilde{I}) ~|~ p = 1, \ldots, N_L \}$.
We will show that the sum $\sum_{p=1}^{N_L} \eta_{\ell,p}^\omega (\tilde{I})$ provides a good approximation to $\xi_L^\omega (\tilde{I})$ in the limit $L \rightarrow \infty$.

\subsection{Existence of infinitely-divisible point measures}\label{subsec:pt-measures1}

We first establish the existence of limit points for the family of local random measures $d \xi^\omega_\Lambda$.
We recall that $\{ E_j^\omega (\Lambda) \}$ is the collection of eigenvalues of the local Hamiltonians $H_\Lambda^\omega$.
We mention that if we had a Minami estimate then we could prove the uniqueness of the limit point. It is not clear that the extended Minami estimate can be used to establish uniqueness. We let $\mathcal{B}(\R)$ denote the set of Borel subsets of $\R$.

\begin{pr}\label{prop:limit-pts1}
Let $E_0$ be in the regime of complete localization. The family of local random point measures
$$
d \xi^\omega_\Lambda (x) = \sum_j \delta( |\Lambda| (E_j^\omega (\Lambda) - E_0) - x) dx,
$$
is tight. Any limit point $\xi^\omega$ of this family is an infinitely-divisible point process.
\end{pr}

\begin{proof}
To prove tightness, we need to show that for any bounded $I \in \mathcal{B}(I)$,
 $$
\lim_{t \rightarrow \infty} \limsup_{|\Lambda| \rightarrow \infty} \P \{ \xi_\Lambda^\omega(I) > t \}
= 0.
$$
This follow from the Wegner estimate, Lemma \ref{lemma:wegner1}, and the Chebychev inequality
 $$
\P \{ \xi_\Lambda^\omega(I)  > t \} \leq \frac{c |I|}{t}.
$$
So there is a random measure $\xi^\omega$ and a sequence $L_n \rightarrow \infty$ so that
$\xi_{\Lambda_{L_n}}^\omega \rightarrow \xi^\omega$ in distribution.
Since the set of point measures on $\R$ is closed in the set of Borel measures on the line, any limit point $\xi^\omega$ is
a point measure. Since these random variables $\xi^\omega(I) \in \Z^+$, these random point measures are also called point processes.
These limit points are infinitely-divisible point measures. This follows from the comparison with the uniformly asymptotically negligible array formed from the local point measures $\eta_{p, \ell}^\omega$ and described in Lemma \ref{lemma:asympt-negl1}.
\end{proof}



\subsection{Analysis of the independent array of point processes}\label{subsec:indep-pt-process1}

We begin with an analysis of the independent array of point processes $\{ \eta_{\ell,p}^\omega ~|~ p=1, \ldots, N_L, ~\ell = 1, 2, \ldots \}$.
We are interested in the limit $L \rightarrow \infty$ of the sum $\sum_{p=1}^{N_L} \eta_{\ell,p}^\omega$.
We recall the definition of $\eta_{p, \ell}^\omega(I)$, for a bounded Borel subset $I \subset \R$, from \eqref{eq:process2}.

\subsubsection{Existence of infinitely-divisible point measures}\label{subsec:array-pt-measures1}

We establish the existence of limit points for the array $\eta_{\ell,p}^\omega$
an a manner analogous to section \ref{subsec:pt-measures1}.

\begin{pr}\label{prop:array-limit-pts1}
Let $E_0$ be in the regime of complete localization. The family of local random point measures
$\zeta_\Lambda^\omega := \sum_{p=1}^{N_L} \eta_{\ell,p}^\omega$
is tight. Any limit point $\zeta^\omega$ of this family is an infinitely-divisible point process.
\end{pr}

\begin{proof}
To prove tightness, we need to show
 $$
\lim_{t \rightarrow \infty} \limsup_{|\Lambda| \rightarrow \infty} \P \{ \zeta_\Lambda^\omega(I) > t \}
= 0.
$$
This follow from the Wegner estimate, Lemma \ref{lemma:wegner1}, and the Chebychev inequality
\bea\label{eq:array-tight1}
\P \{ \zeta_\Lambda^\omega(I)  > t \} & \leq & \frac{1}{t} \E \{ \zeta_\Lambda^\omega (I) \} \nonumber \\
 & \leq  & \frac{N_L}{t}(2 \ell + 1)^d |\tilde{I}| C_W \nonumber \\
 & \leq &  \frac{C_W |I|}{t} \left( \frac{2 \ell+1}{2L+1} \right)^d \left( \frac{ 2L+1}{2\ell+1} \right)^d \nonumber \\
  & \leq & \frac{C_W |I|}{t}.
\eea
So there is a random measure $\zeta^\omega$ and a sequence $L_n \rightarrow \infty$ so that
$\zeta_{\Lambda_{L_n}}^\omega \rightarrow \zeta^\omega$ in distribution.
Since the set of point measures on $\R$ is closed in the set of Borel measures on the line, any limit point $\zeta^\omega$ is
a point measure. It is infinitely-divisible by construction.
\end{proof}

\subsubsection{Asymptotic negligibility}\label{subsec:asym-negl1}

We now turn to an analysis of the limit points $\zeta^\omega$.
We first prove that the array $\{ \eta_{\ell, p}^\omega \}$ is uniformly asymptotically negligible.

\begin{lemma}\label{lemma:asympt-negl1}
For any $E \in\Sigma_{CL}$ and interval $I := [a,b]$, we set $\tilde{I} := \beta_L^{-1} I + E$. Let $\ell$ satisfy $\ell L^{-1} \rightarrow 0$ as $L \rightarrow \infty$. Then for all $\epsilon > 0$,
we have
\beq\label{eq:indep1}
\lim_{L \rightarrow \infty} \sup_{p=1,\ldots,N_L} \P \{ \eta_{\ell,p}^\omega ({I}) > \epsilon \}  = 0 .
\eeq
\end{lemma}

\begin{proof}
This follow from the Wegner estimate and the Chebychev inequality:
\bea\label{eq:asympt-negl1}
\P \{ \eta_{\ell,p}^\omega (\tilde{I}) > \epsilon \} & \leq& \frac{1}{\epsilon}\E \{ {\rm Tr} E_{\tilde{I}}(H_{\Lambda_\ell(p)}^\omega) \}
 \nonumber \\
 & \leq & \frac{c|I|}{\epsilon}\left( \frac{ 2 \ell + 1}{2L+1} \right)^d ,
 \eea
 for some finite constant $c > 0$.
 This upper bound is uniform in the index $p$.
\end{proof}

\subsubsection{Asymptotic support property}\label{subsec:asym-support1}

We next use the extended Minami estimate, Proposition \ref{proposition:minami1}, to
characterize the asymptotic support of the process.

\begin{lemma}\label{lemma:asympt-support1}
For  any $E \in\Sigma_{CL}$ and interval $I := [a,b]$, we set $\tilde{I} := \beta_L^{-1} I + E$. Let $\ell$ satisfy $\ell L^{-1} \rightarrow 0$ as $L \rightarrow \infty$. Then we have
we have
\beq\label{eq:indep2}
\lim_{L \rightarrow \infty} \sum_{p=1}^{N_L} \P \{ \eta_{\ell,p}^\omega ({I}) > m_k \}  = 0 .
\eeq
\end{lemma}

\begin{proof}
From the extended Minami estimate of section \ref{sec:minami1}, we have
\beq\label{eq:minami2}
\P \{ \eta_{\ell,p}^\omega ({I}) > m_k \} \leq  {C_M}  |\tilde{I}|^2 (2 \ell + 1)^{2d} .
\eeq
Consequently, the sum in \eqref{eq:indep2} is bounded above as
\beq\label{eq:indep3}
\sum_{p=1}^{N_L} \P \{ \eta_{\ell,p}^\omega ({I}) > m_k \} \leq C_M |I|^2 \left( \frac{2\ell + 1}{2 L +1} \right)^{d},
\eeq
which vanishes as $L \rightarrow \infty$ under the condition on $\ell$.
\end{proof}

\subsubsection{Distribution of limit random variables for $\zeta_L^\omega = \sum_{p=1}^{N_L} \eta_{\ell,p}^\omega$}\label{subsubsec:limit-indep-case1}

\noindent
We  combine the previous results with Lemma \ref{lemma:fixed-prob1} in order to show that the limit random variables
$\zeta^\omega(I)$, for any limit point $\zeta^\omega$, are distributed according to a compound Poisson point processes.

\begin{pr}\label{prop:limit-indep-case1}
For any $\epsilon > 0$, set $\ell = L^{(1-\epsilon)/2}$. Let $I \in \mathcal{B}(\R)$. There is a sequence $L_n \rightarrow \infty$
so that the sequence of random variables $\zeta_{L_n}^\omega(I) = \sum_{p=1}^{N_{L_n}} \eta_{\ell,p}^\omega(I)$
converges to a random variable $\zeta^\omega (I)$ that is distributed according to
a compound Poisson point process with Levy measure $M$
supported on $\{1, \ldots, m_k\}$
with weights $p_j(\cdot): \mathcal{B}(\R) \rightarrow \R^+$, described in \eqref{eq:fixed-prob3}:
$$
dM(\lambda \times I) = \sum_{j=1}^{m_k} \delta (\lambda - j) p_j(I) ~d\lambda, ~~~\forall I \in \mathcal{B}(\R).
$$
\end{pr}

\begin{proof}
We compute the characteristic function of the sum of independent random variables $\eta_{\ell,p}^\omega ({I})$ as follows:
\bea\label{eq:indep-array1}
\E \left\{ e^{it \sum_{p=1}^{N_L} \eta_{\ell,p}^\omega ({I}) } \right\} & = & \prod_{p=1}^{N_L}
 \E \{ e^{it \eta_{\ell,p}^\omega ({I}) } \}  \nonumber \\
 & = & \prod_{p=1}^{N_L}
  e^{\log [ \E \{ e^{it \eta_{\ell,p}^\omega ({I}) } \}]}  \nonumber \\
   & = &  e^{ \sum_{p=1}^{N_L} \log [ \E \{ e^{it \eta_{\ell,p}^\omega ({I})} - 1 \} + 1 ]}   .
\eea
The expectation in the exponential is estimated by
\beq\label{eq:indep-array2}
|\E \{ e^{it \eta_{\ell,p}^\omega ({I})} - 1 \} | \leq t \E \{  \eta_{\ell,p}^\omega ({I}) \} \leq c t |I| \beta_L^{-1}
(2\ell + 1 )^d .
\eeq
as follows from the Wegner estimate Lemma \ref{lemma:wegner1}.
This justifies the approximation $\log ( 1 + x) = x + \mathcal{O}(|x|^2)$ as $L \rightarrow \infty$, so that
\bea\label{eq:indep-array3}
\E \left\{ e^{it \sum_{p=1}^{N_L} \eta_{\ell,p}^\omega ({I}) } \right\}
& = &
 e^{ \sum_{p=1}^{N_L} \log [ \E \{ e^{it \eta_{\ell,p}^\omega ({I})} - 1 \} + 1 ]} \nonumber \\
& = & e^{ \sum_{p=1}^{N_L} \E \{ e^{it \eta_{\ell,p}^\omega ({I})} - 1 \} + t^2 c^2 N_L (2L +1)^{-2d} }
\eea
For all $t$ fixed, we have $\lim_{L \rightarrow \infty} N_L (2L+1)^{-2d} = \lim_{L \rightarrow \infty} [(2 \ell+1) (2L+1)]^{-d} = 0$, so we will drop this
term from the exponent.
As a consequence, we obtain
\bea\label{eq:indep-array4}
\E \left\{ e^{it \sum_{p=1}^{N_L} \eta_{\ell,p}^\omega ({I}) } \right\}
& = & e^{ \sum_{p=1}^{N_L} \E \{ e^{it \eta_{\ell,p}^\omega ({I})} - 1 \}} \nonumber \\
& =& e^{ \sum_{j=1}^{(2 \ell +1)^d} (e^{itj} - 1) ~\sum_{p=1}^{N_L} \P  \{ \eta_{\ell,p}^\omega ({I}) = j \} }
\eea
We now use the extended Minami estimate, Proposition \ref{proposition:minami1}. We have
\bea\label{eq:asymp-supp3}
\sum_{j=1}^{(2\ell+1)^d} (e^{itj} - 1) ~\sum_{p=1}^{N_L} \P  \{ \eta_{\ell,p}^\omega ({I}) = j \}
&=& \sum_{j=1}^{m_k} (e^{itj} - 1) ~\sum_{p=1}^{N_L} \P  \{ \eta_{\ell,p}^\omega ({I}) = j \}  + R(\ell), \nonumber \\
 & &
\eea
where the remainder $R(\ell)$ is estimated as
\bea\label{eq:asymp-supp4}
| R(\ell) | & \leq & 2 \sum_{p=1}^{N_L} \sum_{j=m_k+1}^{(2 \ell+1)^d} \P  \{ \eta_{\ell,p}^\omega ({I}) = j \}  \nonumber \\
 & \leq  & 2 (2 \ell +1)^d N_L \P  \{ \eta_{\ell,p}^\omega ({I}) \geq m_k+1  \} \nonumber \\
  & \leq & 2 C_M |I|^2 \left( \frac{ (2 \ell + 1)^2}{(2L+1)} \right)^d .
 \eea
 so we may replace the sum over $j$ by the sum up to $m_k$ with a negligible error. Dropping this term as above, we find
\bea\label{eq:indep-array5}
 \E \left\{ e^{it \sum_{p=1}^{N_L} \eta_{\ell,p}^\omega ({I}) } \right\}  &=&
 e^{ \sum_{j=1}^{m_k} (e^{itj} - 1) ~\sum_{p=1}^{N_L} \P  \{ \eta_{\ell,p}^\omega ({I}) = j \} } \nonumber \\
 & = & e^{ \sum_{j=1}^{m_k} (e^{itj} - 1) \P  \{ \cup_{p=1}^{N_L} \eta_{\ell,p}^\omega ({I}) = j \} }
 \eea
where we used the independence of the random variables $\eta_{\ell,p}^\omega ({I})$.
From Lemma \ref{lemma:fixed-prob1}, there exists a subsequence $L_n$ so that we have
$$
\lim_{L_n \rightarrow \infty} \P  \{ \cup_{p=1}^{N_L} \eta_{\ell,p}^\omega ({I}) = j \} = p_j(I) .
$$
This result, and \eqref{eq:indep-array5}, proves the proposition.
\end{proof}



\subsection{Relation between the two processes}\label{subsec:relation1}

We establish the relationship between $\xi_L^\omega ({I})$ and $\zeta_L^\omega (I) = \sum_{p=1}^{N_L} \eta_{\ell,p}^\omega ({I})$. Since $E \in \Sigma_{\rm CL}$, we can control the difference of these processes in the large $L$ regime using the localization bounds.

%


\begin{pr}\label{theorem:relation1}
The random point measures $\zeta_L^\omega = \sum_{p=1}^{N_L} \eta_{p, \ell}^\omega$ and $\xi_L^\omega$ have the same limit points in the sense of distributions as $L \rightarrow \infty$.
\end{pr}

\begin{proof}
As in Minami's paper \cite[section 2]{minami1}, we compare the Laplace transforms of these measures. Since $| e^{-X} - e^{-Y}| \leq |X-Y|$,
we have
\beq\label{eq:laplace-trans1}
\left| \E \left\{ e^{- \xi_L^\omega (f)} - e^{- \sum_{p=1}^{N_L} \eta_{p, \ell}^\omega (f)} \right\} \right|
  \leq  \E | \xi_L^\omega (f) - \sum_{p=1}^{N_L} \eta_{p, \ell}^\omega (f) | .
\eeq
for a good class of functions $f$. The details of the proof are similar to those in Minami \cite{minami1}.
\end{proof}


\subsubsection{Proof of Lemma \ref{lemma:fixed-prob1}}\label{subsubsec:proof-main1}

\begin{proof}
The first part of Lemma \ref{lemma:fixed-prob1} follows from the Wegner estimate, Lemma \ref{lemma:wegner1}. The existence of the limit of the expectation and its value are proved exactly as in the proof of equation (2.50) of Minami's paper \cite{minami1} using Proposition \ref{theorem:relation1}
to replace $\xi_L^\omega(\tilde{I})$ with $\sum_{p=1}^{N_L} \eta_{p,\ell}^\omega(\tilde{I})$, so we omit it.
The second part \eqref{eq:fixed-prob1} follows directly from Lemma \ref{lemma:asympt-support1}.
The last result \eqref{eq:fixed-prob3} follows from the uniform bound \eqref{eq:fixed-prob0}.
Finally at least one of the $p_j(I)$ is non-zero since
$$
\sum_{j=1}^{m_k} j p_j(I) = \lim_{L\rightarrow \infty} {\mathbb E}(\xi_L^\omega({I})) = n_{m_k}(E) |I| \neq 0.
$$
\end{proof}

\subsubsection{Proof of Theorem \ref{thm:main1}}\label{subsubsec:proof-main2}

\begin{proof}
Let $\xi^\omega$ be a limit point of $\xi_L^\omega$ described in Proposition \ref{prop:limit-pts1} so there exists a sequence
$L_n \rightarrow \infty$ so that $\xi_{L_n}^\omega \rightarrow \xi^\omega$ in the distributional sense. Because of this, we may use the result
described after Proposition \ref{prop:pt-measures1} in the appendix to conclude that the limit in \eqref{eq:fixed-prob3}
exists for the sequence $\{ L_n \}$. then, the proof of Theorem \ref{thm:main1} is obtained by combining Proposition \ref{theorem:relation1} and
Proposition \ref{prop:limit-indep-case1} to show that the characteristic functions
$$
{\mathbb E}(e^{it\xi_{L_n}^\omega({I})})
$$
have the limit
$$
e^{\sum_{j=1}^{m_k} (e^{itj} -1) p_j(I)}.
$$
This proves Theorem \ref{thm:main1}.
\end{proof}


\section{Eigenvalue point processes in the continuous case}\label{sec:continuous1}
\setcounter{equation}{0}

We now consider random Schr\"odinger operators of the form $H_\omega = L + V_\omega$ on $L^2 (\R^d)$
where $L = - \Delta$, the Laplacian on $\R^d$, and $V_\omega$ is the Anderson-type random potential given in \eqref{eq:potential1-cont}.
We prove that for $E \in \Sigma_{\rm CL}$, the regime of complete localization defined in section \ref{sec:introduction}, the local eigenvalue statistics in each fixed interval is a compound Poisson for which the L\'evy measure has support in the set of positive integers.

We will prove this result using the Levy-Khintchine Theorem \cite[Theorem 1.2.1]{applebaum}.
A random variable is said to be infinitely-divisible if its distribution function is infinitely-divisible.
Let us recall that if $X$ is a non-negative random variable with characteristic function expressed as
\beq\label{eq:charac-exp1}
\E \{ e^{i \lambda X} \} = e^{- \Psi (\lambda)},
\eeq
then the function $\Psi (\lambda)$ is called the characteristic exponent of $X$.
The L\'evy-Khintchine formula characterizes infinitely-divisible random variables as
random variables whose characteristic exponent $\Psi$ has the form
\beq\label{eq:charac-exp2}
\Psi(\lambda) = i\lambda b - a \lambda^2 + \int_\R (e^{i\lambda x} -1 - i \lambda x \chi_{[-1,1]}(x)) dM(x)
\eeq
where $dM(x)$ is a Borel measure on $\R - \{ 0 \}$ satisfying $\int_\R ~(1 \wedge x^2) dM(x) < \infty$, and $\chi_B$
is the characteristic function for the set $B \subset \R$. The measure $M$
is called the L\'evy measure of $X$. It is clear that if
$\Psi$ is a bounded function of $\lambda$, then $a,b=0$ in \eqref{eq:charac-exp1}. Hence, the random variable
is compound Poisson distributed with the Levy measure $M$ (see, for example, Item 4 in the Notes following
Theorem 1.2.1 of \cite{applebaum}).

Our main result for random Schr\"odinger operators on $\R^d$ is the following theorem.
This characterization of the random variables requires only the Wegner estimate and localization. Localization is
necessary to prove the infinite divisibility of the random variables.

\begin{thm}\label{thm:main-continum1}
Let $H^\omega$ be an  Anderson model \eqref{eq:model1} on $\R^d$.
Let $E \in \Sigma_{\rm{CL}}$.
The limit points $\xi^\omega$ of the eigenvalue point processes $\xi_L^\omega$
defined in \eqref{eq:process1} are infinitely-divisible point processes.
For each bounded Borel subset $I \subset \R$, the characteristic functions
$$
{\mathbb E} (e^{it\xi_L^\omega(I)})
$$
of the random variables $\xi_L^\omega(I)$ have limit points of the form
\beq\label{eq:conv-process112}
e^{\int (e^{its} - 1) dM_I (s)},
\eeq
where the L\'evy measure $M_I$ has support in the set of positive integers. Hence, the
vague limit points of random variables $\xi_L^\omega(I)$ are
compound Poisson distributed with the associated L\'evy measure supported in the set $\N$.
\end{thm}


\begin{proof}
\noindent
1. We begin by noting that if we let $\zeta_L^\omega ({I}) := \sum_{p=1}^{N_L} \eta_{\ell,p}^\omega({I})$, as above, then we have
$$
\E \{ e^{it \xi_L^\omega ({I})} \} = \E \{ e^{it \zeta_L^\omega ({I})} \} + \E \{ e^{it \xi_L^\omega ({I})} - e^{it \zeta_L^\omega ({I})} \}.
$$
We estimate the second term on the right as
\beq\label{eq:weak-conv1}
\| \E \{ e^{it \xi_L^\omega ({I})} - e^{it \zeta_L^\omega ({I})} \} \| \leq t \E \{ | \zeta_L^\omega ({I}) - \xi_L^\omega ({I}) | \}.
\eeq
The weak convergence of the processes, as proved in section 6 of \cite{CGK2} shows that the limit as $L \rightarrow \infty$ in \eqref{eq:weak-conv1} is zero. This requires only the Wegner estimate and the decay estimates on the Green's functions as follows from the fact that $E \in \Sigma_{\rm CL}$. Repeating the arguments of section \ref{subsec:array-pt-measures1}, we establish the existence of limit points for $\zeta_L^\omega$.
As in the lattice case, the result \eqref{eq:weak-conv1} shows that $\xi_L^\omega(I)$ and $\zeta_L^\omega(I)$ has the same weak limit points.
These weak limit points are infinitely-divisible point processes.

\noindent
2. Next, we analyze the characteristic exponent of $\zeta_L^\omega (I)$.
Proceeding as in the proof of Proposition \ref{prop:limit-indep-case1}, it follows from \eqref{eq:indep-array3}
that
\beq\label{eq:indep-array6}
\E \left\{ e^{it \sum_{p=1}^{N_L} \eta_{\ell,p}^\omega ({I}) } \right\}
 =  e^{ \sum_{p=1}^{N_L} \E \{ e^{it \eta_{\ell,p}^\omega ({I})} - 1 \}  } ,
\eeq
 up to a term vanishing as $L \rightarrow \infty$.
Hence, we can assume that the characteristic exponent for $\xi_L^\omega (I)$ is
\beq\label{eq:charact-exp3}
\Psi_L(t) = \sum_{p=1}^{N_L} \E \{ e^{it \eta_{\ell,p}^\omega ({I})} - 1 \} =  N_L \E \{ e^{it\eta_{\ell,1}(I)} - 1 \},
\eeq
using the homogeneity in the index $p$.


\noindent
3. We now choose a sequence $\{ L_k \}$ and an infinitely-divisible point process $\xi^\omega$ so that $\xi_{L_k}^\omega(I) \rightarrow \xi^\omega(I)$ in distribution. It follows that $\Psi_{L_k}(t) \rightarrow \Psi(t)$, and, because $\xi_L^\omega(I)$ and $\zeta_L^\omega(I)$
have the same limit points,
\bea\label{eq:charact-exp31}
\lim_{k \rightarrow \infty} \E ( e^{- t \zeta_{L_k}^\omega(I)} ) & = & \lim_{k \rightarrow \infty} e^{- \Psi_{L_k}(t)} \nonumber \\
  &=& \E ( e^{it \xi^\omega(I)}) \nonumber \\
  &=& e^{- \Psi (t)} ,
\eea
where $\Psi_{L_k}(t)$ is given in \eqref{eq:charact-exp3}.

\noindent
4. We next prove a uniform bound on $\Psi_L(t)$. Since $\eta_{\ell,1}^\omega({I})$ is the trace of a projection, it is integer-valued.
The subset where $\eta_{\ell,1}^\omega({I}) =0$ does not contribute to $\Psi_L(t)$ since
$e^{\eta_{l,1}^\omega({I})} -1 = 0$ there. Hence, this observation and the Chebychev inequality imply
\bea\label{eq:psi-bound1}
\Psi_L(t) &=& N_L \E \{ (e^{it\eta_{\ell,1}(I)} - 1) \chi_{\eta_{\ell,1}^\omega ({I}) \geq 1} \} \nonumber \\
 & \leq & 2 N_L \P \{ \eta_{\ell,1}^\omega ({I}) \geq 1 \} \nonumber \\
  & \leq & 2 N_L \E \{ \eta_{\ell,1}^\omega ({I}) \} .
  \eea
The expectation is estimated using the Wegner estimate. Consequently, we have the bound
\beq\label{eq:psi-bound2}
\Psi_L(t) \leq 2 \left(\frac{2L+1}{2\ell+1} \right)^d \frac{|I|(2 \ell+1)^d}{(2L+1)^d} \leq 2 |I|,
\eeq
uniform in $L$. Hence, $\Psi(t)$ is bounded.

\noindent
5. We write $\Psi_L(t)$ as
\beq\label{eq:charact-exp4}
\Psi_L(t) = \sum_{j=1}^\infty (e^{itj} - 1) \P ( \zeta_L^\omega (I)=j),
\eeq
and note that by the Wegner estimate and Chebychev inequality,
\beq\label{eq:charact-exp5}
\P ( \zeta_L^\omega (I) = j)  \leq \frac{1}{j} \E ( \zeta_L^\omega(I))  \leq \frac{C_W}{j}|I|.
 \eeq
Consequently, we can find a subsequence $\{ L_m \}$ so that
\beq\label{eq:charact-exp6}
\lim_{m \rightarrow \infty} \P ( \zeta_{L_m}^\omega (I) = j) = p_j(I).
\eeq
As a consequence, choosing another subsequence, if necessary, we have
\bea\label{eq:charac-exp3}
\lim_{k \rightarrow \infty} e^{- \Psi_{L_k}(t)}  &=& \lim_{k \rightarrow \infty} \E ( e^{- t \zeta_{L_k}^\omega(I)} ) \nonumber \\
   &=& e^{- \Psi (t)} \nonumber \\
   &=& e^{\sum_{j=1}^\infty (e^{itj} - 1) p_j(I)} .
   \eea
This proves that $\xi^\omega (I)$ is distributed according to a compound Poisson process with L\'evy measure supported on $\N$ with weights $p_j(I)$.

\end{proof}


We remark that a Minami estimate for continuous models would help us better characterize the L\'evy measure.

\section{Examples of random operators with non-Poisson statistics}\label{sec:non-poisson}
\setcounter{equation}{0}

We present two examples of random operators for which $\xi^\omega(I)$ is distributed according to a compound Poisson distribution that is not Poissonian. Both show that the multiplicity of the eigenvalues
of the local operators $H_\Lambda^\omega$ has a direct effect on the distribution of
$\xi^\omega(I)$.

\begin{myexample}
We take the operators given in \eqref{eq:model1} with $L=0$ and ${\rm rank}~(P_i) = 2$ with
the single site distribution $\mu$ absolutely continuous with derivative $n$.
We can think of this model as the infinite disorder limit of the generalized Anderson model of the type
considered in \eqref{eq:model1}, by putting a disorder parameter $h$
in front of $L$ and setting $h$ to zero.
It is clear
that $\mu$ is the IDS for this model and $n$ is its density of states.  The spectrum of $H^\omega$
is pure point almost surely with compactly supported eigenvectors.
Let $\Sigma(H^\omega)$ denote the almost sure spectrum of $H^\omega$. We then have $\Sigma_{\rm CL} = \Sigma (H^\omega)$.
The set of eigenvalues of $H^\omega$ is given by
$$
\mathrm{eigenvalues}(H^\omega) = \{ \omega_i : i \in {\mathbb Z}^d \} .
$$
Since the rank of $P_i$ is 2, each eigenvalue has multiplicity 2.
similarly, the finite set of eigenvalues of the local operators
$H^\omega_\Lambda$  are given by
$$
\sigma(H_{\Lambda}^\omega) = \{ \omega_i : i \in \Lambda \},
$$
and each eigenvalue has multiplicity $2$.


Turning to the random variables $\xi_\Lambda^\omega$, for any $E$, writing $\tilde{I} := |\Lambda|^{-1} I + E$, we have
$$
{\rm Tr}(E_{H_\Lambda^\omega}(\tilde{I})) = 2 |\{ i \in \Lambda : \omega_i \in \tilde{I} \}|
$$
hence it is always an \textbf{even} integer (including zero). We take
$E \in \Sigma(H^\omega)$ with $n(E) \neq 0$.
Using the independence of the random variables $\omega_j$ and the definition of the measures $p_j(I)$ in \eqref{eq:fixed-prob3},
we easily compute the measures $p_j(I)$ of Theorem \ref{thm:main1} and find
$$
p_{2}(I) = n(E)|I|, ~~{\rm and} ~~~ p_{j}(I) =0, j \neq 2.
$$
Therefore, by the remark after Theorem \ref{thm:main1} concerning the characterization of a Poisson distribution,
these limiting random variables are compound Poisson, but not Poisson, distributed.
\end{myexample}

\noindent
The second example is a family of random Schr\"odinger operators with kinetic energy term given by the discrete Laplacian $L$.

\begin{myexample}
Consider the operator $H^\omega$ as in \eqref{eq:model1} on $\ell^2({\mathbb Z}^d)\otimes{\mathbb C}^{m_k}$ with
$$
P_i = |\delta_i\rangle\langle \delta_i|\otimes I_{m_k} , ~~ i \in \Z^d,
$$
where $|\delta_i\rangle\langle \delta_i|$ is the projector onto the site $i \in \Z^d$, and
$I_{m_k}$ is the identity matrix on ${\mathbb C}^{m_k}$.  Then clearly all the
eigenvalues of $H_\Lambda^\omega$ for any finite $\Lambda \subset {\mathbb Z}^d$ have uniform multiplicity
$m_k$. Arguing as in the previous example, we find that the Levy measure
has the form as in \eqref{eq:conv-process11}. In particular, if $n(E)$ is the non-zero density of states of the operator
$L + \sum_{j\in {\mathcal J}} |\delta_i\rangle\langle \delta_i| \omega_i$, the weights $p_i(I)$ are given by
$$
p_{m_k}(I) = n(E) |I|, ~~ p_{j}(I) = 0, ~~ j \neq m_k,
$$
showing that the eigenvalue statistics is strictly non-Poisson but compound Poisson distribution.
Naboko, Nichols, and Stolz \cite{naboko-nichols-stolz} considered a similar model with $I_{m_k}$ replaced by a diagonal matrix $W =
 {\mbox diag} ~( \lambda_1, \ldots, \lambda_{m_k})$, with $\lambda_j > 0$. they showed that if the eigenvlaues of $W$ are all simple, then

\end{myexample}


\section{Appendix: Convergence of measures and Lemma \ref{lemma:fixed-prob1}}\label{sec:convergence1}
\setcounter{equation}{0}

We state a proposition on the convergence of measures.

\begin{pr}\label{prop:pt-measures1}
Suppose $\mu_n$ is a sequence of locally finite (finite)  measures, converging to a locally finite
measure $\mu$ vaguely (weakly).  Suppose further that for all $n$,
$$
\mathrm{supp} (\mu_n), \mathrm{supp} (\mu) \subseteq \mathcal{S} \subset \R
$$
where $\mathcal{S}$ is a discrete subset of $\R$.  Then
$$
\lim_{n\rightarrow \infty} \mu_n(\{s\}) = \mu(\{s\}), ~~ \mathrm{for ~ all} ~ s \in \mathcal{S}.
$$
\end{pr}

\begin{proof}
Since $\mathcal{S}$ is discrete, it is countable and hence the measures $\mu_n, \mu$ are atomic.
Therefore their distribution functions
$$
\Phi_n(x) = \begin{cases} \mu_{n}((0, x]), x >0 \\ \mu_n((x, 0]), ~ x < 0 \end{cases},
\Phi(x) = \begin{cases} \mu((0, x]), x >0 \\ \mu((x, 0]), ~ x < 0 \end{cases},
$$
satisfy
$$
\Phi_n(s+\delta) - \Phi_n(s-\delta) = \mu_n(\{s\}), ~~
\Phi(s+\delta) - \Phi(s-\delta) = \mu(\{s\}) ~~ \mathrm{for ~ some} ~~ \delta >0,
$$
$\delta$ chosen such that $(s-\delta, s+\delta) \cap \mathcal{S} = \{s\}$.
Since $\mu_n$ converges to $\mu$ vaguely (weakly) we also have
$$
\lim_{n} \Phi_n(y) = \Phi(y),
$$
for every point of continuity $y$ of $\Phi$ and hence by definition every $y \in \R \setminus \mathcal{S}$.
Therefore,
$$
\lim_{n} \mu_n(\{s\}) = \lim_{n} (\Phi_n(s+\delta) - \Phi_n(s-\delta)) =
 (\Phi(s+\delta) - \Phi(s-\delta)) = \mu(\{s\}).
$$
\end{proof}

We apply this proposition to prove the convergence property used in the proof of Theorem \ref{thm:main1} in section \ref{subsubsec:proof-main2}.
We fix a bounded Borel set $I$ and apply the above proposition to the measures
$$
\mu_L (\{j\}) = \P\{ \xi_L^\omega({I}) = j\}
= \P\circ\xi_L^{\omega}({I})^{-1}(\{j\}), ~~
\mu (\{j\}) = \P\{ \xi^\omega(I) = j\}.
$$
Given a limit point $\xi^\omega$ of the family $\xi_L^\omega$ as in Proposition \ref{prop:limit-pts1}, there is a sequence $L_m \rightarrow \infty$ so that $\xi_{L_m}^\omega \rightarrow \xi^\omega$. Hence,
the random variables $\xi_{L_m}^\omega({I})$ converge to $\xi^\omega(I)$ in distribution
which means that the distribution of $\mu_{L_m}$
converge to the distribution $\mu$. All these measures have their support in $\Z^+$.
Therefore, by Proposition \ref{prop:pt-measures1}, we conclude that
$$
\lim_{m \rightarrow \infty} \P(\xi_{L_m}^\omega({I}) = j\}) = p_j(I) = \P(\{\xi^\omega(I) = j\}).
$$

%

\end{document}